\documentclass{llncs}
\usepackage{color}
\usepackage[T1]{fontenc}
\usepackage[utf8]{inputenc}
\usepackage{amsmath,amssymb} 
\usepackage{hyperref}
\usepackage{bbm}
\usepackage{ifthen}
\newcommand{\dfn}[1]{\emph{#1}}
\newcommand\subfini{\Subset}
\usepackage{stmaryrd}

\newcommand{\Z}{\mathbbm Z}
\newcommand{\N}{\mathbbm N}
\newcommand{\M}{\mathbbm M}

\newcommand{\A}{\mathcal A}
\newcommand{\B}{\mathcal B}
\newcommand{\G}{\mathcal G}
\newcommand{\forb}{\mathcal F}
\newcommand{\lang}{\mathcal L}
\newcommand{\supp}[1]{\mathcal S(#1)}
\newcommand{\Aut}[1]{\mathcal Aut(#1)}
\newcommand{\End}[1]{\mathcal End(#1)}
\newcommand{\WP}[1]{\mathcal{WP}(#1)}

\newcommand{\ie}{\textit{i.e.}\ }
\newcommand{\compl}[1]{{#1}^C}
\DeclareMathOperator*{\id}{id}
\newcommand\ito\rightarrowtail
\newcommand\oto\twoheadrightarrow

\newcommand{\defeq}{:=}
\newcommand\prenum{\preceq_e}
\newcommand\eqenum{\equiv_e}

\newcommand{\sett}[2]{\left\{\left.#1\vphantom{#2}\right|#2\right\}}
\newcommand{\set}[3]{\sett{#1\in#2}{#3}}
\newcommand\card[1]{\left|#1\right|}
\title{Undecidable word problem in subshift automorphism groups}
\author{Pierre Guillon\inst1 \and Emmanuel Jeandel\inst2 \and Jarkko Kari\inst3 \and Pascal Vanier\inst4}
\institute{Universit\'e d'Aix-Marseille, CNRS, Centrale Marseille\\  I2M, UMR 7373 -- 13453 Marseille, France
	\\\email{pierre.guillon@math.cnrs.fr}
	\and Université de Lorraine, CNRS, Inria\\ LORIA -- F 54000 Nancy, France
	\\\email{emmanuel.jeandel@loria.fr}
	\and Department of Mathematics and Statistics\\ FI-20014 University of Turku, Finland
	\\\email{jkari@utu.fi}
	\and
  Laboratoire d'Algorithmique, Complexit\'e et Logique\\
  Universit\'e de Paris-Est, LACL, UPEC, France
  \\\email{pascal.vanier@lacl.fr}
 }
\begin{document}
\maketitle
\begin{abstract}
  This article studies the complexity of the word problem in groups of
  automorphisms of subshifts. We show in particular that for any Turing degree,
  there exists a subshift whose automorphism group contains a subgroup whose
  word problem has exactly this degree.
\end{abstract}
Subshifts are sets of colorings of a group $G$ avoiding some family of forbidden
patterns. The most commonly studied kind of subshifts are the subshifts of
finite type (SFTs), which when $G=\Z^2$ correspond to sets of Wang tilings, up
to a recoding. They have been introduced independently as a way of discretizing
dynamical systems on compact spaces and as a tool to study decidability
questions. They have also been used as models for complex systems.

An automorphism of a subshift $X$ is a shift-invariant continuous bijection from
$X$ to $X$, or equivalently a reversible cellular automaton on $X$. Little is
known about automorphism groups of subshifts in general, besides that they are
countable, for instance it is not known whether the automorphism groups of the
2-symbol full shift and of the 3-symbol full shift are isomorphic.

However some properties of the base subshift influence its
automorphism group: for instance when a multidimensional subshift has positive
entropy it contains every finite group~\cite{hochman_2010}, but this is not a
necessary condition.

In dimensions $d\geq 2$, computability has played a central role in the study of
SFTs, sofic and effective shifts. From a computability point of view, it is known
that the word problem in the automorphism group of an SFT is computably
enumerable. Here we show that for any given computably enumerable degree, one can construct
an SFT the automorphism group of which has a word problem with this degree. 
 
\section{Preliminaries}
By countable set, we mean injectable in $\N$.
Let $\lambda$ denote the empty word.
For $\A$ a countable alphabet, we note $\A^*\defeq\bigsqcup_{n\in\N}\A^n$ 
 the set of finite words 
 over $\A$. We also note $\A^{\le r}\defeq\bigsqcup_{n\le r}\A^n$ for $r\in\N$.
 
Let us note $\compl X$ the complement of set $X$. $W\subfini X$ means that $W\subset X$ and $W$ is finite. $V\sqcup W$ means $V\cup W$ assuming that $V\cap W=\emptyset$.
 
\subsection{Computability}
Computability problems are naturally defined over $\N$, but can easily be extended through subsets of it, cartesian products or disjoint union (by canonically injecting $\N$ in sets of tuples). For example, if $\G\subset\N$, then the set $\G^*$ of tuples admits a simple injection into $\N$.
Let us fix a (computable) countable set $I$, that we can identify to integers.

\begin{definition}
Let us define the following reducibilities, for $X,Y\subset I$:
\begin{enumerate}
\item 
 $X$ is \dfn{Turing-reducible} to $Y$, $X\leq_T Y$, if: one can compute $X$ with oracle $Y$.
\item 
 $X$ is \dfn{enumeration-reducible} to $Y$, $X\leq_eY$, if: from any $x$ and any integer $i\in\N$, one can compute a finite set $Y_i(x)$ such that $x\in X$ if and only if $\exists i\in\N,Y_i(x)\subset Y$.
\item 
 $X$ is \dfn{positive-reducible} to $Y$, $X\leq_pY$, if: from any $x$, one can compute finitely many finite sets $Y_0(x),\ldots,Y_{n-1}(x)$ such that $x\in X$ if and only if $\exists i<n,Y_i(x)\subset Y$.
\item 
 $X$ is \dfn{many-one-reducible} to $Y$, $X\leq_m Y$, if: from any $x$, one can compute some $\phi(x)$ such that $x\in X$ if and only if $\phi(x)\in Y$.
\item 
 $X$ is \dfn{one-one-reducible} to $Y$, $X\leq_1 Y$, if, $X\leq_m Y$ and $\phi$ is one-to-one.
\end{enumerate}\end{definition}
One-one reducibility implies many-one reducibility, which in turns implies positive-reducibility, which implies both Turing-reducibility and enumeration-reducibility.

Each reducibility $\leq_r$ induces a notion of equivalence $\equiv_r$:
$A\equiv_r B$ iff $A\leq_r B$ and $B\leq_r A$. And each notion of equivalence
$\equiv_r$ induces a notion of degree $\deg_r$: the \dfn{degree of a set} $A$ is its
equivalence class for $\equiv_r$.

The join $A\oplus B$ of $A$ and $B$ is the set $C$ such that $2n+1\in C$ iff
$n\in A$ and $2n\in C$ iff $n\in B$. It has the property that $A\leq_r A\oplus
B$ and $B\leq_r A\oplus B$ for any reducibility $\leq_r$ previously defined.


See \cite{rogers} for a reference on computability-theoretical reductions.

\subsection{Monoids and groups}
We will deal with countable monoids $\M=\G^*/R$, where $\G\subset\N$, $\G^*$ is the free monoid generated by symbols from $\G$ and $R$ is a monoid congruence\footnote{We could deal in the same way with semigroups, by prohibiting the empty word.}.
The monoid is always implicitly endowed with its generating set $\G$ (later, some problems may depend on the presentation).
Each element of the monoid is represented by a word $u\in\G^*$, but the representation is not one-to-one (except for the free monoid itself).
We note $i=_\M j$ if $\pi(i)=\pi(j)$ and $\pi:\G^*\to\M$ is the natural quotient map.

\newcommand{\restr}[1]{_{\left|#1\right.}}

It is also clear that the concatenation map, which from any two words $i,j\in\G^*$ outputs $i\cdot j$, which is one representative of the corresponding product, is computable. 
We say that $\M$ is an \dfn{effective group} if, additionnally, there is a computable map $\psi:\G^*\to\G^*$ such that $i\cdot\psi(i)=_\M\psi(i)\cdot i=_\M\lambda$.

The \dfn{equality problem} of $\M$, endowed with generating family $\G$, is the set of pairs $\set{(i,j)}{(\G^*)^2}{i=_\M j}$, endowed with a natural enumeration so that we can consider it as a computability problem.

\begin{remark}\label{r:subg}~\begin{enumerate}
\item\label{i:word} It is clear that the \dfn{word problem} $\set i{\G^*}{i=_\M\lambda}$ is one-one-reducible to the equality problem.
\item If $\M$ is an effective group, then the word problem is actually many-one-equivalent to the equality problem.
\item The equality problems for $\M$ endowed with two distinct finite generating sets are one-one-equivalent.
\item If $\M'$ is a submonoid of $\M$ endowed with a generating set which is included in that of $\M$, then the equality problem in $\M'$ is one-one-reducible to that of $\M$.
\item\label{i:subg} In particular, the equality problem in any finitely generated submonoid is one-one-reducible to that of $\M$.
\end{enumerate}\end{remark}
Nevertheless, there are countable groups whose word problem is computable when endowed with one generating family, and uncomputable when endowed with another one.

The word problem is known to be decidable if and only if the group is \dfn{computable} (see \cite{MR0113807} for a proof in the finitely generated case), that is, it can be seen as a computable subset of $\N$ over which the composition rule is a computable function (this implies that inversion is also a computable map).


\subsection{Subshifts}
Let $\A$ be a finite alphabet with at least two letters, and $\M$ a group (most of the following should be true if $\M$ is a cancellative monoid though).
A finite \dfn{pattern} $w$ over $\A$ with \dfn{support} $W=\supp w\subfini\G^*$ is a map $w=(w_i)_{i\in W}\in\A^W$.
Depending on the context, note that, for $g\in\supp w$, $w_g$ may either be an element of $\A$ or a subpattern with support $\{g\}$.
If $g\in\G^*$ and $w$ is a pattern, we will note $\sigma^g(w)$ the pattern with support $W\cdot g$ such that $\sigma^g(w)_{i\cdot g^{-1}}=w_i$ for all $i\in\supp w$.

We are interested in $\A^\M$, which is a Cantor set, when endowed with the prodiscrete topology,
on which $\M$ acts continuously by (left) shift: we note $\sigma^i(x)_j=x_{i.j}$ for $i,j\in\M$ and $x\in\A^\M$.

A \dfn{subshift} is a closed $\sigma$-invariant subset $X\subset\A^\M$.
Equivalently, $X$ can be defined as the set $X_\forb\defeq\set x{\A^\M}{\forall i\in\M,\forall w\in\forb,\exists j\in\supp w,x_{i\cdot j}\ne w_j}$ avoiding a language $\forb\subset\bigsqcup_{W\subfini\G^*}\A^W$, which is then called a (defining) \dfn{forbidden language}.
If $\forb$ can be chosen finite, the subshift is called \dfn{of finite type} (SFT); if it can be chosen computably enumerable, it is called \dfn{effective}.

The \dfn{language} with \dfn{support} $W\subfini\G^*$ of subshift $X$ is the set $\lang_W(X)\defeq\sett{(x_{\pi(i)})_{i\in W}}{x\in X}$; the \dfn{language} of $X$ is $\lang(X)=\bigsqcup_{W\subfini\G^*}\lang_W(X)$, and its \dfn{colanguage} is the complement of it.
The latter is a possible defining forbidden language. 
If $u\in\lang_W(X)$, we define the corresponding \dfn{cylinder} $[u]=\set xX{\forall i\in W,x_{\pi(i)}=u_i}$.

\begin{remark}\label{r:colang}
$\pi$ induces a natural covering $\Pi:\A^\M\to\A^{\G^*}$ by $\Pi(x)_i=x_{\pi(i)}$.
Its image set $\Pi(\A^\M)$ is a subshift over the free monoid.
One can note the following.
\begin{enumerate}
\item $X=X_{\compl{\lang(X)}}$.
\item\label{i:colfull} The colanguage of the full shift $\A^\M$ is the same as
  that of the subshift $\Pi(\A^\M)$: the set
  \[\compl{\lang(\A^\M)}=\bigsqcup_{W\subfini\G^*}\set w{\A^W}{\exists i,j\in
      W,i=_\M j,w_i\ne w_j}\]
  of patterns that do not respect the monoid congruence.
\item Nevertheless, $\emptyset$ is a forbidden language defining $\A^\M$.
\item\label{i:colang} The colanguage of every subshift $X_\forb\subset\A^\M$ is the set of patterns $w\in\A^W$, $W\subfini\G^*$, whose all extensions to configurations $x\in[u]$ involve as a subpattern a pattern of either $\forb$, or $\compl{\lang(\A^\M)}$.
In that case, by compactness, at least one such subpattern appears within a finite support $V\subfini\G^*$, with $W\subset V$, which depends only on $W$.
\end{enumerate}\end{remark}
\begin{remark}Let $\M$ be a monoid.
\begin{enumerate}
\item The equality problem in $\M$ is positive-equivalent (and one-one-reducible) to the colanguage of the full shift.
\item The colanguage of any subshift $X$ is enumeration-reducible to the join of any defining forbidden language for $X$ and the equality problem of $\M$.
\end{enumerate}\end{remark}
\begin{proof}~\begin{enumerate}
\item one-one-reducibility: one can computably map each word $(i,j)\in(\G^*)^2$ to a unique pattern over $\{i,j\}$ involving two different symbols.
By Point~\ref{i:colfull} of Remark~\ref{r:colang}, this pattern is in the colanguage of the full shift if and only if $i=_\M j$.
\\
positive-reducibility (with all $Y_i$s being singletons): from each pattern $w\in\A^{\G^*}$, one can compute the set of pairs $(i,j)\in\supp w^2$ such that $w_i\ne w_j$.
By Point~\ref{i:colfull} of Remark~\ref{r:colang}, $w$ is in the colanguage if and only if one of these pairs is an equality pair in $\M$.
\item
Consider the set $Z$ of \emph{locally inadmissible} patterns, that involve a subpattern either from the forbidden language or from $\compl{\lang(\A^\M)}$.
From any pattern $w$, one can enumerate all of its subpatterns and all of their shifts, \ie all patterns $v$ such that there exists $i\in\G^*$ with $\supp v\cdot i\subset\supp w$ and $w_{j\cdot i}=v_j$ for every $j\in\supp v$.
This shows that $Z$ is enumeration-reducible to the join of the forbidden language and $\compl{\lang(A^\M)}$, the latter being equivalent to the equality problem, by the previous point.
It remains to show that the colanguage of $X$ is enumeration-reducible to $Z$.
\\
From any pattern $w\in\A^{\G^*}$ and any $i\in\N$, one can compute some $V_i\subfini\G^*$ including $\supp w$, in a way that $V_{i+1}\supset V_i$ and $\bigcup_{i\in\N}V_i=\G^*$ (for example take the union of $\supp w$ with balls in the Cayley graph).
Then, one can compute the set $Y_i$ of extensions of $w$ to $V_i$, \ie patterns with support $V_i$ whose restriction over $\supp w$ is $w$.
By Point~\ref{i:colang} of Remark~\ref{r:colang}, $w\in\compl{\lang(X)}$ if and only if there exists $V\subfini\G^*$ with $V\supset\supp w$ such that all extensions of $w$ to $V$ are in $Z$; and in particular this should happen for some $V_i$, which precisely means that $Y_i\subset Z$.
\qed\end{enumerate}\end{proof}

It results that, in some sense, one expects most subshifts to have a colanguage at least as complex as the equality problem in the underlying monoid.

\subsection{Homomorphisms}
Let $X\subset\A^\M$ and $Y\subset\B^\M$ be subshifts.
Denote $\End{X,Y}
$ the set of \dfn{homomorphisms} (continuous shift-commuting maps) from $X$ to $Y$, 
 and $\Aut{X,Y}
$ the set of bijective ones (\dfn{conjugacies}).
We also note $\End X=\End{X,X}$ the monoid of \dfn{endomorphisms} of $X$, 
 and $\Aut X=\Aut{X,X}$ the group of its \dfn{automorphisms}.


If $\M$ is finitely generated, then homomorphisms correspond to block maps (or cellular automata), thanks to a variant of the Curtis-Hedlund-Lyndon theorem \cite{hedlund}.
\begin{theorem}
Let $\M$ be finitely generated.
A map $\Phi$ from subshift $X\subset\A^\M$ into subshift $Y\subset\B^\M$ is a homomorphism if and only if there exist a \dfn{radius} $r\in\N$ and a \dfn{block map} $\phi:\A^{\G^{\le r}}\to\B$ such that for every $x\in\A^\M$ and $i\in\G^*$, $\Phi(x)_{\pi(i)}=\phi(x\restr{\pi(i\cdot\G^{\le r})})$ (where the latter has to be understood with the obvious reindexing of the argument).
\end{theorem}
Let us order the block maps $\phi:\A^{\G^{\le r}}\to\B$ by increasing radius $r\in\N$, and then by lexicographic order, so that we have a natural bijective enumeration $\N\to\bigsqcup_{r\in\N}\B^{\A^{\G^{\le r}}}$ (because $\A$, $\B$ and $\G$ are finite).
This gives in particular a surjective enumeration $\N\to\End{\A^{\G^*},\B^{\G^*}}$ and in general, a partial surjective enumeration $\N'\subset\N\to\End{X,Y}$.
In general, $\N'\ne\N$. 
It is a nontrivial problem to ask whether $\N'$ is computable (this is the case for the full shift when $\M=\Z$), but not the topic of the present paper.
Obtaining a bijective enumeration for $\End{\A^{\G^*},\B^{\G^*}}$ would be easily achieved by enumerating each block map only for its smallest possible radius.
Nevertheless, trying to achieve a bijective enumeration in general for $\End{X,Y}$, or even for $\End{\A^\M,\B^\M}$, is a process that would depend on the colanguage of the subshift (we want to avoid two block maps that differ only over the colanguage), which may be uncomputable.

Even when $\M$ is an effective group, $\Aut X$ need not be an effective group!

For the rest of the paper, let us assume that $\M$ is an effective group.
More precisely, all results could be interpreted as reductions to a join with a problem representing the composition map of the group, and sometimes to an additional join with a problem representing the inversion.

\section{Equality problem is not too hard}
\begin{remark}~\label{r:reduc}
Two distinct block maps $\phi,\psi:\A^{\G^{\le r}}\to\A$ representing an endomorphism of $X$ actually represent the same one if and only if for every pattern $u\in\A^{\G^{\le r}}$, $\phi(u)\ne\psi(u)\Rightarrow u\in\compl{\lang(X)}$.
\end{remark}

The equality problem is at most as complex as the language.
\begin{theorem}\label{t:endlang}
The equality problem in $\End X$ is positive-reducible to $\compl{\lang(X)}$.
\end{theorem}
\begin{proof}
One can directly apply Remark~\ref{r:reduc}, by noting that it is easy to transform each block map into an equivalent one, so that the resulting two block maps have the same radius (the original maximal one, by ignoring extra symbols).
\qed\end{proof}
Of course, this remains true for the equality problem in $\Aut X$. 
Since positive-reducibility implies both Turing-reducibility and enumeration-reducibility, we get the following for the lowest classes of the arithmetic hierarchy (which was already known; see \cite{hochman_groups_nodate}).
\begin{corollary}~\begin{enumerate}
\item The equality problem is decidable, in the endomorphism monoid of any subshift with computable language (for instance 1D sofic subshift, 1D substitutive subshift, minimal effective subshift, two-way space-time diagrams of a surjective cellular automaton\ldots). 
\item The equality problem is computably enumerable, in the endomorphism monoid of any effective subshift (for instance multidimensional sofic subshift, substitutive subshift, limit set of cellular automaton\ldots).
\end{enumerate}\end{corollary}

\section{Automorphism groups with hard equality problem}
The purpose of this section is to prove a partial converse to Theorem~\ref{t:endlang}: a subshift $X$ for which the two problems involved are equivalent, however complex they are.
%
%
%


Let $X\subset\A^\M$ and $Y\subset\B^\M$ be subshifts.
For $\alpha:\B\to\B$ and $u\in\A^\M$, let us define the \dfn{controlled map} $C_{u,\alpha}$ as the homomorphism over $X\times Y$ such that $C_{u,\alpha}(x,y)_0=(x_0,\alpha(y_0))$ if $x\in[u]$; $(x_0,y_0)$ otherwise.
Denote also $\pi_1$ the projection to the first component, and $\sigma^g_1$ the shift of the first component with respect to element $g\in\M$: $\sigma^g_1(x,y)_0=(x_g,y_0)$ for every $(x,y)\in X\times Y$.
\pagebreak\begin{remark}\label{r:controlled}~\begin{enumerate} 
\item $\pi_1C_{u,\alpha}=\pi_1$.
\item If $\M$ is a group and $g\in\M$, then $C_{u,\alpha}=\sigma^g_1C_{\sigma^{g^{-1}}(u),\alpha}\sigma^{-g}_1$.
\item $C_{u,\alpha}\in\End{X\times Y,X\times\B^\M}$.
\item $C_{u,\alpha}$ is injective if and only if $\alpha$ is a permutation.
\item $C_{u,\alpha}\in\End{X\times Y}$ if $Y$ is (locally) \dfn{$\alpha$-permutable}, \ie for all $y\in Y$, if we define $z$ by $z_0=\alpha(y_0)$, $z_i=y_i$ for $i\ne0$, then $z\in Y$.
\item\label{i:id} From Remark \ref{r:reduc}, $C_{u,\alpha}$ is the identity over $X\times Y$ if and only if $u\notin\lang(X)$ or $\alpha$ is the trivial permutation over letters appearing in $Y$.
\end{enumerate}\end{remark}
\begin{example}
Examples of $\alpha$-permutable subshifts are the full shift on $\B$ or, if $\B=\B'\sqcup\{\bot\}$ and $\alpha(\bot)=\bot$, the \dfn{$\B'$-sunny-side-up} defined by forbidding every pattern which involves two occurences of $\B'$.
We have seen that the colanguage of the former is positive-equivalent to the word problem in $\M$.
The language of the latter can be easily proven to be many-one-equivalent to the word problem in $\M$ (as essentially noted in \cite[Prop 2.11]{aubrun_notion_2017}), hence yielding a kind of jump for the colanguage.
\end{example}

If $a,b,c\in\B$, let us denote $\alpha_{abc}$ the \dfn{$3$-cycle} mapping $a$ to $b$, $b$ to $c$, $c$ to $b$, and any other element to itself.
The following lemma corresponds essentially to \cite[Lemma~18]{boykett_finite_2017}.
\begin{lemma}\label{l:decomp}
Suppose $\B$ has at least 5 distinct elements $a,b,c,d,e$.
Let $u\in\A^{\supp u}$ be a pattern, $g\in\supp u$, and $v=u\restr{\supp u\setminus\{g\}}$.
Then $C_{u,\alpha_{abc}}=(\Psi\Phi)^2$, where $\Phi=\sigma^g_1C_{u_g,\alpha_{ade}}C_{u_g,\alpha_{bad}}\sigma^{g^{-1}}_1$ and $\Psi=C_{v,\alpha_{bde}}C_{v,\alpha_{cbd}}$.
\end{lemma}
\begin{proof}
If $x_g=u_g$, then $\Phi(x,y)_0=(x_0,\phi(y_0))$, where $\phi$ is the involution that swaps $a$ and $b$ on the one hand, $d$ and $e$ on the other hand; otherwise $\Phi(x,y)_0=(x_0,y_0)$.
If $x\in[v]$, then $\Psi(x,y)_0=(x_0,\psi(y_0))$, where $\psi$ is the involution that swaps $b$ and $c$ on the one hand, $d$ and $e$ on the other hand; otherwise $\Psi(x,y)_0=(x_0,y_0)$.
Since $\phi^2=\psi^2=\id$, one can see that if $x\notin[u]$, then $(\Psi\Phi)^2(x,y)_0=(x_0,y_0)$.
Now if $x\in[u]$, then we see that $\Psi\Phi(x,y)_0=(x_0,\psi\phi(y_0))$, and $\psi\phi=\alpha_{acb}$, so that we get the stated result.
\qed\end{proof}

\begin{theorem}\label{t:langredwp}
Let $X\subset\A^\M$ be a subshift and $Y\subset\B^\M$ an $\alpha_{abc}$-permutable subshift for every $a,b,c\in\B'\subset\B$, where $\card{\B'}\ge5$.
Then $\compl{\lang(X)}$ is one-one-reducible to the word problem in the subgroup of automorphisms of $X\times Y$ generated by $\sigma^g_1$ and $C_{u_0,\alpha_{abc}}$ for $g\in\G$, $a,b,c\in\B'$ and $u_0\in\A$.
\end{theorem}
\begin{proof}
From an induction and Lemma \ref{l:decomp}, we know that this subgroup includes every $C_{u,\alpha_{abc}}$ for every $a,b,c\in\B'$ and $u\in\A^*$.
From Point \ref{i:id} of Remark \ref{r:controlled}, an automorphism $C_{u,\alpha_{abc}}$ is equal to the identity if and only if $u\notin\lang(X)$.
\qed\end{proof}
Consequently, subshifts can have finitely generated groups with equality problem as complex as their colanguage, as formalized by the following corollary.
In that case, the equality problem of the whole automorphism group is as complex also.
\begin{corollary}\label{c:vrac}~\begin{enumerate} 
\item If $X$ and $Y$ are as in Theorem~\ref{t:langredwp}, then $\compl{\lang(X)}$ is one-one-equivalent to the word problem in (a finitely generated subgroup of) $\Aut{X\times Y}$. 
\item For every subshift $X$ over a finitely generated group $\M$, there exists a countable-to-one extension $X\times Y$ such that $\compl{\lang(X)}$ is one-one-equivalent to the word problem in (a finitely generated subgroup of) $\Aut{X\times Y}$.
\item For every subshift $X$ over a finitely generated group $\M$, there exists a full extension $X\times\B^\M$ such that $\compl{\lang(X)}$ is one-one-equivalent to the word problem in (a finitely generated subgroup of) $\Aut{X\times\B^\M}$.
\item Every $\Sigma^0_1$ Turing degree contains the word problem in (a finitely generated subgroup of) $\Aut X$, for some 2D SFT $X$.
\item\label{i:uncomput} There exists a 2D SFT $X$ for which the word problem in (a finitely generated subgroup of) $\Aut X$ is undecidable.
\end{enumerate}
\end{corollary}
Point \ref{i:uncomput} answers \cite[Problem 5]{hochman_groups_nodate}.
\begin{proof}~\begin{enumerate}
\item Just use Point~\ref{i:subg} of Remark~\ref{r:subg}.
For the converse reduction in the one-one-equivalence, simply apply Theorem~\ref{t:endlang} and Point~\ref{i:word} of Remark~\ref{r:subg}.
\item We use Theorem~\ref{t:langredwp} with $Y$ being the $\{0,1,2,3,4\}$-sunny-side-up.
\item We use Theorem~\ref{t:langredwp} with $Y=\{0,1,2,3,4\}^\M$.
Remark that $\compl{\lang(X)}$ and $\compl{\lang(X\times\{0,1,2,3,4\}^\M)}$ are one-one-equivalent.
\item Every $\Sigma^0_1$ degree contains the colanguage of a 2D SFT, thanks to constructions from \cite{MEdv,Myers}.
Then its product with the full shift $\{0,1,2,3,4\}^{\Z^2}$ is still an SFT, and we conclude by the previous point.
\item Apply the previous point with any uncomputable $\Sigma^0_1$ degree.
\qed\end{enumerate}\end{proof}

Note that the number of generators can be decreased if we want to reduce only the language whose support is spanned by a subgroup.
For instance 2D SFTs are already known to have (arbitrarily $\Sigma^0_1$) uncomputable 1D language.
Indeed, our automorphisms do not alter the $X$ layer, so that their parallel applications to all traces with respect to a subgroup is still an automorphism.

Among the open questions, we could wonder whether there is a natural class of SFT (irreducible, with uncomputable language, at least over $\Z^2$) whose colanguage could be proven reducible to the word problem in the automorphism group. This could require to encode the whole cartesian product of Theorem~\ref{t:langredwp} inside such subshifts.
Another question would be to adapt our construction while controling the automorphism group completely so that it is finitely generated.

\subsection*{Acknowledgements}
This research supported by the Academy of Finland grant 296018.

We thank Ville Salo for some discussions on commutators, on the open questions, and for a very careful reading of this preprint.

\bibliographystyle{unsrt}
\bibliography{wordpb}

\begin{thebibliography}{1}

\bibitem{hochman_2010}
Michael Hochman.
\newblock On the automorphism groups of multidimensional shifts of finite type.
\newblock {\em Ergodic Theory and Dynamical Systems}, 30(3):809–840, 2010.

\bibitem{rogers}
Hartley Rogers, Jr.
\newblock {\em Theory of Recursive Functions and Effective Computability}.
\newblock MIT Press, Cambridge, MA, USA, 1987.

\bibitem{MR0113807}
Michael~O. Rabin.
\newblock Computable algebra, general theory and theory of computable fields.
\newblock {\em Transactions of the American Mathematical Society}, 95:341--360,
  1960.

\bibitem{hedlund}
Gustav~Arnold Hedlund.
\newblock Endomorphisms and automorphisms of the shift dynamical system.
\newblock {\em Mathematical Systems Theory}, 3:320--375, 1969.

\bibitem{hochman_groups_nodate}
Michael Hochman.
\newblock Groups of automorphisms of {SFT}s.
\newblock Open problems ;
  \url{http://math.huji.ac.il/~mhochman/problems/automorphisms.pdf}.

\bibitem{aubrun_notion_2017}
Nathalie Aubrun, Sebastián Barbieri, and Mathieu Sablik.
\newblock A notion of effectiveness for subshifts on finitely generated groups.
\newblock {\em Theoretical Computer Science}, 661:35--55, 2017.

\bibitem{boykett_finite_2017}
Tim Boykett, Jarkko Kari, and Ville Salo.
\newblock Finite generating sets for reversible gate sets under general
  conservation laws.
\newblock {\em Theoretical Computer Science}, 701:27--39, November 2017.

\bibitem{MEdv}
Stephen~G. Simpson.
\newblock {M}edvedev degrees of 2-dimensional subshifts of finite type.
\newblock {\em Ergodic Theory and Dynamical Systems}, 34(November
  2012):665--674, 2014.

\bibitem{Myers}
William Hanf and Dale Myers.
\newblock Non recursive tilings of the plane {II}.
\newblock {\em Journal of Symbolic Logic}, 39(2):286--294, 1974.

\end{thebibliography}
\end{document}